\documentclass{amsart}

\usepackage{enumerate}
\usepackage{amssymb}
\usepackage{amsfonts}
\usepackage{latexsym}
\usepackage{amsmath}
\usepackage{euscript}
\usepackage{graphics}
\usepackage{graphicx}
\usepackage{tikz}
\usetikzlibrary{spy}
\graphicspath{ {./images/} }

\newtheorem{theorem}{Theorem}[section]
\newtheorem{proposition}[theorem]{Proposition}

\newtheorem{lemma}[theorem]{Lemma}
\newtheorem{definition}[theorem]{Definition}
\theoremstyle{remark}

\newtheorem{remark}[theorem]{Remark}

\numberwithin{equation}{section}

\def\diag{{\rm diag\,}}

\def\RR{\mathbb{R}}

\def\CC{\mathbb{C}}

\def\tr{\textrm{Tr}}

\def\1{\mathbb{1}}

\newmuskip\pFqmuskip

\newcommand*\pFq[6][8]{%
  \begingroup 
  \pFqmuskip=#1mu\relax
  \mathcode`\,=\string"8000
  \begingroup\lccode`\~=`\,
  \lowercase{\endgroup\let~}\pFqcomma
  {}_{#2}F_{#3}{\left[\genfrac..{0pt}{}{#4}{#5};#6\right]}%
  \endgroup
}
\newcommand{\pFqcomma}{\mskip\pFqmuskip}

\usepackage[colorinlistoftodos]{todonotes}

\title[Jacobi matrices that realize PST and ESE]{Jacobi matrices that realize perfect quantum state transfer and Early State Exclusion}

\author[R.~Bailey]{Rachel~Bailey}
\address{
RB,
Department of Mathematical Sciences\\
Bentley University\\
175 Forest Street\\
Waltham, MA 02452, USA}
\email{rbailey@bentley.edu}

\author[S.~Costa]{Sara~Costa}
\address{
SC,
Department of Mathematics\\
University of Hartford\\
200 Bloomfield Avenue\\
West Hartford, CT 06117, USA}
\email{costa.a.sara33@gmail.com}

\author[M.~Derevyagin]{Maxim~Derevyagin}
\address{
MD,
Department of Mathematics\\
University of Connecticut\\
341 Mansfield Road, U-1009\\
Storrs, CT 06269-1009, USA}
\email{maksym.derevyagin@uconn.edu}

\author[C.~Findley]{Caleb~Findley}
\address{
CF,
Department of Mathematics\\
University of Texas Arlington\\
701 S. Nedderman Drive\\
Arlington, TX 76019, USA}
\email{scf0853@mavs.uta.edu}

\author[K.~Zuang]{Kai~Zuang}
\address{
KZ,
Department of Mathematics\\
Brown University\\
151 Thayer Street\\
Providence, RI 02912, USA}
\email{kaizuang@gmail.com}

\subjclass{Primary 15A29, 81P45; Secondary 47B36, 33C45}
\keywords{Jacobi matrices, inverse problems, quantum information, Krawtchouk polynomials, Chebyshev polynomials}

\begin{document}

\begin{abstract}
In this paper we show how to construct 1D Hamiltonians, that is, Jacobi matrices, that realize perfect quantum state transfer and also have the property that the overlap of the time evolved state with the initial state is zero for some time before the transfer time. If the latter takes place we call it an early exclusion state. We also show that in some case early state exclusion is impossible. The proofs rely on  properties of Krawtchouk and Chebyshev polynomials. 
\end{abstract}

\maketitle

\section{Introduction}

In recent years there has been interest in studying various mathematical aspects of perfect quantum state transfer \cite{T1}, \cite{DDMT}, \cite{Kay10}, \cite{KMPPZ}, \cite{VZh12}, and \cite{T2}. To give a quick account of what this is from the mathematical perspective, let us consider a Jacobi matrix $J$, that is, a symmetric tridiagonal matrix of the form
\[
J=\begin{pmatrix}
		a_0 & b_0    &  &  \\
		b_0 & a_1    & \ddots &  \\
		& \ddots & \ddots & b_{N-1} \\
		&        & b_{N-1}& a_{N}     \\
	\end{pmatrix},
	\]
where $a_k$'s are real numbers and $b_k$'s are positive numbers. Then $J$ represents the Hamiltonian of the system of $N+1$ particles with the nearest-neighbor interactions. The evolution of the system is given by $e^{-itJ}$. We say that $J$ realizes Perfect State Transfer (or, shortly, PST) if there exist $T_0>0$ and $\phi\in\RR$ such that
\begin{equation}\label{PST}
e^{-iT_0J}{\bf e}_0=e^{i\phi}{\bf e}_N,
\end{equation}
where ${\bf e}_0$ and ${\bf e}_N$ are the first and last elements of the standard basis 
$\{ {\bf e}_k\}_{k=0}^N$ of $\CC^{N+1}$. In other words, this means that in such a system, qubits get transferred from site $1$ to site $N+1$ in time $T_0$. Note that we can only know probabilities of quantum states and \eqref{PST} describes the situation of quantum transfer with probability $1$, which is a desirable but rare phenomenon. Nevertheless, theoretically one can easily construct such Hamiltonians since \eqref{PST} is equivalent to the two conditions (see \cite{Kay10} or \cite{VZh12})
\begin{enumerate}
\item[(i)] $J$ is persymmetric, that is,
\[
\begin{split}
a_{k}=a_{N-k},\quad k=0,1,\dots, N, \\
 b_{l}=b_{N-1-l}, \quad l=0,1,\dots, N-1;    
\end{split}
\]
\item[(ii)] the eigenvalues $\lambda_0<\lambda_1<\dots<\lambda_N$ of $J$ satisfy
\begin{equation}\label{PST_eigenvalues}
    \lambda_{k+1}-\lambda_k=\frac{(2n_k+1)\pi}{T_0}, \quad k=0,...,N,
\end{equation} 
where $n_k$ are nonnegative integers. 
\end{enumerate}
This characterization gives some flexibility for engineering quantum wires. Once such a wire is built, one can potentially attempt to speed up the transfer of several qubits by initiating the second transfer at time $t$ earlier than $T_0$. This could be possible if we know that the overlap of the time evolved state with the initial state is zero for some time before the transfer time. Thus it makes sense to introduce the following concept. 

\begin{definition}(Early State Exclusion)
Let J be a Jacobi matrix that has earliest perfect state transfer at time $T_0$. If there is a time $t$ such that  $0<t<T_0$ and
\[
(e^{-iJt}{\bf e}_0,{\bf e_0})_{\CC^{N+1}} = 0
\]
then we say that J has $\textbf{Early State Exclusion}$, (ESE), at time $t$.
\end{definition}
By looking at the above definition, it is not clear if such a $J$ exists and in fact the third author has learned that the question whether such $J$ exists is open from Christino Tamon who attributed the question to Alastair Kay. In this note we show that the answer to this question is positive and construct an infinite family of such Hamiltonians for any odd $N\ge 3$.

\section{Trial and error}

In this section we will consider the situation of Early State Exclusion for matrices of sizes $2$, $3$, and $4$, which corresponds to $N=1$, $N=2$, and $N=3$, respectively.

\subsection{The case of $2\times 2$-matrices.} Based on the characterization given in the previous section, any symmetric, persymmetric  matrix
\[
J=\begin{pmatrix}
a_0&b_0\\
b_0&a_0
\end{pmatrix}, \quad  a_0\in \mathbb{R}, \quad b_0> 0
\]
realizes a PST since in this case we can always find $T_0$ such that $\lambda_1-\lambda_0=\pi/T_0$, where $\lambda_1, (>)\, \lambda_0$ are the eigenvalues of $J$. However, since we are looking at the zeroes of $(e^{-iJt}{\bf e}_0,{\bf e}_0)$ we can simplify the form of $J$ further. Namely, without loss of generality we can assume that the matrix $J$ has the form
\[
J=\begin{pmatrix}
0&1\\
1&0
\end{pmatrix}
\]
Indeed, the shift $J\to J-a_0I$ simply brings a factor of $e^{ia_0t}$, which is never 0, to the function in question and then rescaling $b_0$ to $1$ only changes the transfer time $T_0$. Next, one can find that
\[
e^{-iJt}{\bf e}_0=e^{-it\begin{pmatrix}
0&1\\
1&0
\end{pmatrix}}\begin{pmatrix}
1\\
0
\end{pmatrix}=\begin{pmatrix}
\cos t\\
-i\sin t
\end{pmatrix}.
\]
The latter relation shows that if $(e^{-iJt}{\bf e}_0,{\bf e_0})_{\CC^{N+1}}=0$ then the magnitude of the second component of $e^{-iJt}{\bf e}_0$ is $1$, which in turn implies that PST takes place at time $t$. As a result, $2\times 2$ Jacobi matrices do not have ESE.
\subsection{The case of $3\times 3$-matrices.}
Here we need to do a bit more computations and so we will formulate the result first.
\begin{proposition}
    Let $J$ be a $3\times 3$ Jacobi matrix such that PST occurs for the first time at time $T_0$. 
    Then J does not have ESE.
\end{proposition}

\begin{proof}
Assume that there exists $t<T_0$ such that $(e^{-iJt}{\bf e}_0,{\bf e}_0)_{\CC^{N+1}} = 0$. Note that since PST occurs, $J$ must be persymmetric and hence takes the form
\begin{equation*}
J = \begin{pmatrix}
a_0 & b_0 & 0 \\
b_0 & a_1 & b_0 \\
0 & b_0 & a_0
\end{pmatrix}
\end{equation*} 
for $a_0,a_1\in \mathbb{R}$ and $b_0>0$. As in the case of $2\times 2$-matrices, without loss of generality one can assume that $J$ has the form
\begin{equation*}
J = \begin{pmatrix}
0 & 1 & 0 \\
1 & c & 1 \\
0 & 1 & 0
\end{pmatrix}.
\end{equation*}
Let $\lambda_0<\lambda_1<\lambda_2$ be the eigenvalues of $J$. Note that since we assume PST occurs, the distinctness of the eigenvalues follows from \eqref{PST_eigenvalues}. It is evident that $\det J=0$ and so one of the eigenvalues of $J$ must be $0$. Moreover, since $(J{\bf e}_0,{\bf e}_0)=0$ and ${\bf e}_0\not\in\ker J$, the matrix $J$ must have positive and negative eigenvalues. Taking into account \eqref{PST_eigenvalues} we get that
\[
\lambda_1=-(2m+1)\frac{\pi}{T_0},\quad \lambda_2=0,\quad \lambda_3=(2n+1)\frac{\pi}{T_0}
\]
for some nonnegative integers $n$, $m$. Consequently, we get that
\[
c=\tr \,J=2(n-m)\frac{\pi}{T_0}
\]
and therefore
\begin{equation*}
J = \begin{pmatrix}
0 & 1 & 0 \\
1& \frac{2n \pi - 2n \pi }{T_0} & 1 \\
0 & 1& 0
\end{pmatrix}.
\end{equation*}
Now we can compute $e^{-iJt}$ using Sylvester's formula,
\begin{equation*}
   e^{-iJt}=
   \frac{e^{-i\lambda_1t}}{\lambda_1(\lambda_1-\lambda_3)}J(J-\lambda_3I)+\frac{1}
   {\lambda_1\lambda_3}(J-\lambda_1I)(J-\lambda_3I)+\frac{e^{-i\lambda_3t}}{(\lambda_3-\lambda_1)\lambda_3}(J-\lambda_1I)J,
\end{equation*}
where we took into account that $\lambda_2=0$. We do not need to compute the entire matrix and, in order to get to the function that we want to analyze, we just need to know that
\[
J(J-\lambda_3I){\bf e}_0=\begin{pmatrix}
1\\
\lambda_1\\
1
\end{pmatrix}, (J-\lambda_1I)(J-\lambda_3I){\bf e}_0=\begin{pmatrix}
\lambda_1\lambda_3+1\\
0\\
1
\end{pmatrix},
(J-\lambda_1I)J{\bf e}_0=\begin{pmatrix}
1\\
\lambda_3\\
1
\end{pmatrix}.
\]
The latter yields
\[
(e^{-iJt}{\bf e}_0,{\bf e_0})=\frac{e^{-i\lambda_1t}}{\lambda_1(\lambda_1-\lambda_3)}+
\frac{\lambda_1\lambda_3+1}{\lambda_1\lambda_3}+\frac{e^{-i\lambda_3t}}{(\lambda_3-\lambda_1)\lambda_3},
\]
or, for simplicity, we set 
\[
(e^{-iJt}{\bf e}_0,{\bf e_0})=\alpha e^{-i\lambda_1t}+\beta +\gamma{e^{-i\lambda_3t}},
\]
and notice that $\alpha$, $\beta$, $\gamma$ are positive numbers. Since we know that the transfer time is $T_0$, it implies $(e^{-iJT_0}{\bf e}_0,{\bf e_0})=0$ and thus we have that
\[
\beta=\alpha+\gamma.
\]
Due to our assumption there is a $t$ such that $0<t<T_0$ and 
\[
\alpha e^{-i\lambda_1t}+\beta +\gamma{e^{-i\lambda_3t}}=0,
\]
which, according to the triangle inequality, implies that $e^{-i\lambda_1t}=e^{-i\lambda_3t}=-1$. Furthermore, we get that 
\[
e^{-iJt}{\bf e}_0=\begin{pmatrix}
0\\
0\\
-1
\end{pmatrix},
\]
which shows that PST takes place at time $t<T_0$ and this contradicts our assumption.
\end{proof}

\subsection{The case of $4\times 4$-matrices.} It is not possible to reproduce the previous arguments in this case and so one has to consult Mathematica. After some computations, one can get the feeling that this case is not hopeless and after some more computations one can produce an example of a Jacobi matrix that realizes PST and has ESE. As a product of such experiments, let us consider the following Jacobi matrix 
\begin{equation}\label{4by4example}
J=\begin{pmatrix}
0&\sqrt{15}/2&0&0\\
\sqrt{15}/2&0&1&0\\
0&1&0&\sqrt{15}/2\\
0&0&\sqrt{15}/2&0
\end{pmatrix}.
\end{equation}
Let us denote the components of the evolved state by $x_j(t)$, that is,
\[
e^{-iJt}{\bf e}_0=\begin{pmatrix}
x_0(t)\\
x_1(t)\\
x_2(t)\\
x_3(t)
\end{pmatrix}.
\]
and plot $x_0$ and $x_3$ using Mathematica, see Figure \ref{Fig1}. Notice that Figure \ref{Fig1} demonstrates that the Jacobi matrix $J$ defined by \eqref{4by4example} gives an example of a $4\times 4$ Jacobi matrix with PST and ESE. We prove this statement rigorously below.
\begin{figure}[h!]
\includegraphics[width=\linewidth]{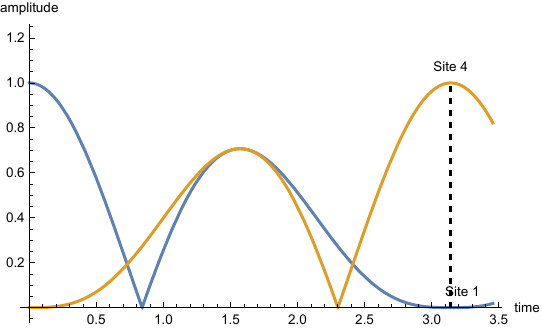}
\caption{This picture demonstrates the occurrence of ESE between $t=0.5$ and $t=1$ and PST at time $t=\pi$.}
\label{Fig1}
\end{figure}

\begin{proposition}
 There is a $4\times 4$ Jacobi matrix $J$ such that it realizes PST and has ESE.
\end{proposition}
\begin{proof}Let $J$ be the $4\times 4$ matrix defined by \eqref{4by4example}. Then the eigenvalues of $J$ are 
\[
-\frac{5}{2},\quad -\frac{3}{2},\quad \frac{3}{2}, \quad \frac{5}{2}
\]
and hence the transfer time $T_0=\pi$ due to \eqref{PST_eigenvalues}. Then, we compute 
\begin{equation}\label{eq:x_0for4by4}
 x_0(t)=\cos^3\left(\frac{t}{2}\right)(3\cos(t)-2),\quad
x_3(t)=-i\sin^3\left(\frac{t}{2}\right)(3\cos(t)+2),   
\end{equation}
from which we clearly see that if $3\cos(t)-2=0$ then 
\[
\sin^2\left(\frac{t}{2}\right)=\frac{1-\cos(t)}{2}=\frac{1}{6} \Rightarrow |x_3(t)|=
\frac{4}{6^{3/2}}<\frac{4}{4^{3/2}}=\frac{1}{2}<1,
\] 
which confirms that $J$ has ESE at $t=\arccos(2/3)<\pi$.
\end{proof}

\section{Absence of Early State Exclusion for Equidistant spectra}\label{sec:absence of ESE}
In this section we will show that a Jacobi matrix that realizes PST and whose spectrum is equidistant, that is, the quantity $\lambda_{k+1}-\lambda_k$ is a constant that does not depend on $k$, cannot have ESE. This fact is based on the properties of  Krawtchouk polynomials and, thus, we start with their basics.

For $p\in(0,1)$ and a positive integer $N$, the monic    
\textit{Krawtchouk polynomials} in $x$ are defined by
\[
K_n(x;p,N)=\sum_{j=0}^n\frac{(-n)_j(-N+j)_{n-j}}{j!}p^{n-j}(-x)_j, \quad n=0,1,\dots, N, N+1,
\]
where
\[(x)_n =\begin{cases}
   1,& n=0,\\
   x(x+1)\dots (x+n-1), & n=1,2,\dots
   \end{cases}
   \] 
is the Pochhammer symbol. Observe that
\[
K_{-1}(x;p,N) = 0, \quad K_0(x;p,N)=1,\quad K_{N+1}(x;p,N)=x(x-1)(x-2)\dots(x-N).
\]
In what follows we will only need the case $p=1/2$ and so let us set $K_n(x): =K_n(x;1/2,N)$. Then the monic Krawtchouk polynomials $K_n(x)$ satisfy the following three-term recurrence relation
   \begin{equation*}
   xK_n(x)=K_{n+1}(x)+\frac{N}{2}K_n(x)+\frac{(N+1-n)n}{4}K_{n-1}(x)
   \end{equation*} 
for $n=0,1,\dots, N$, which can be rewritten in the matrix form
\begin{equation}\label{eq:matrixrec}
x\begin{pmatrix}
K_0(x)\\
K_1(x)\\
\vdots\\
K_{N-1}(x)\\
K_N(x)
\end{pmatrix}=
\begin{pmatrix}
		N/2 & 1   &  &  \\
		{N}/4&  \ddots   & \ddots &  \\
		& \ddots & \ddots & 1 \\
		&        & {N}/4& N/2    \\
	\end{pmatrix}
\begin{pmatrix}
K_0(x)\\
K_1(x)\\
\vdots\\
K_{N-1}(x)\\
K_N(x)
\end{pmatrix}+
\begin{pmatrix}
0\\
0\\
\vdots\\
0\\
K_{N+1}(x)
\end{pmatrix}.
\end{equation}
The latter can be symmetrized by multiplying by the diagonal matrix
\[
D=\diag(1,c_1^{-1/2}, (c_1c_2)^{-1/2}, \dots, (c_1c_2\dots c_N)^{-1/2}),
\] 
where 
\[
c_{k}=\frac{(N+1-k)k}{4}, \quad k=1,2,\dots N,
\]
on the left and by introducing the normalized Krawtchouk polynomials 
$\widehat{K}_n(x)$:
\[
\widehat{K}_0(x):={K}_0(x),\quad \widehat{K}_n(x):=\frac{K_n(x)}{\sqrt{c_1}\sqrt{c_2}\dots \sqrt{c_{n}}},\quad n=1,\dots, N.
\]
In addition, we can make the diagonal of the underlying Jacobi matrix vanish by making the shift:
\[
\mathcal{X}(x):=(\widehat{K}_0(x+N/2),\dots,\widehat{K}_N(x+N/2))^{\top}=D({K}_0(x+N/2),\dots,{K}_N(x+N/2))^{\top}.
\]
Then \eqref{eq:matrixrec} takes the form 
\begin{equation}\label{eq:symmatrixrec}
x\mathcal{X}(x)=J\mathcal{X}(x)+(c_1c_2\dots c_N)^{-1/2}K_{N+1}(x+N/2){\bf e}_N,
\end{equation}
where 
\begin{equation}\label{eq:KrawtchoukJacobi}
J=\begin{pmatrix}
		0 & \sqrt{N}/2    &  &  \\
		\sqrt{N}/2&  \ddots   & \ddots &  \\
		& \ddots & \ddots & \sqrt{N}/2 \\
		&        & \sqrt{N}/2& 0    \\
	\end{pmatrix}.
\end{equation}
From \eqref{eq:symmatrixrec} we see that the zeroes
\[
   x_s = s - \frac{N}{2}, \quad s = 0, 1, 2, ..., N
\]
of $K_{N+1}(x+N/2)$ are the eigenvalues of the Jacobi matrix $J$ given by \eqref{eq:KrawtchoukJacobi} and therefore the vectors $\mathcal{X}(x_s)$ are the eigenvectors of $J$. Now we are in the position to prove the desired result.

\begin{theorem}\label{thrm:equidistantevals}
Assume that a persymmetric Jacobi matrix J of order $N+1$ realizes PST and that its eigenvalues satisfy
\[
\lambda_1-\lambda_0=\lambda_2-\lambda_1=\dots=\lambda_N-\lambda_{N-1},
\]
then ESE is impossible. 
\end{theorem}

\begin{proof} 
Applying the same argument we used in Section 2, we can assume without loss of generality that 
\[
\lambda_1-\lambda_0=\lambda_2-\lambda_1=\dots=\lambda_N-\lambda_{N-1}=1
\]
In addition, as before, the choice of $\lambda_0$ essentially functions as a phase factor for the unitary matrix in question. As such, we can freely set $\lambda_0$ = $- \dfrac{N}{2}$, which yields
$$\lambda_s=x_s = s - \frac{N}{2}, \quad s = 0, 1, 2, ..., N.$$
Since $J$ realizes PST, it is persymmetric. As a result, it coincides with the Jacobi matrix $J$ given by \eqref{eq:KrawtchoukJacobi} due to the Hochstadt uniqueness theorem \cite[Theorem 3]{H74}.   
 Then, we can write 
$${\bf e}_0 = \sum_{s = 0}^{N} f_s\frac{\mathcal{X}(x_s)}{\|\mathcal{X}(x_s)\|}, \quad
f_s=\left({\bf e}_0,\frac{\mathcal{X}(x_s)}{\|\mathcal{X}(x_s)\|}\right).$$
Clearly, we have that
\[
f_s=\frac{1}{{\|\mathcal{X}(x_s)\|}}\left({\bf e}_0,{\mathcal{X}(x_s)}\right)=\frac{1}{{\|\mathcal{X}(x_s)\|}}K_0(x_s+N/2)=\frac{1}{{\|\mathcal{X}(x_s)\|}}.
\]
As a result we get that
$${\bf e}_0 = \sum_{s = 0}^{N}\frac{\mathcal{X}(x_s)}{\|\mathcal{X}(x_s)\|^2}.$$
Next, by the Christoffel-Darboux formula we obtain that
\[
\|\mathcal{X}(x_s)\|^2=\sum_{l=0}^N\widehat{K}_l(s)^2=h_N^{-1}{\widehat{K}_N(s) K_{N+1}'(s)}, \quad s=0,1,\dots, N,
\]
where $h_N$ is a positive constant, whose value is not important for now. Consequently, we arrive at the following relation
 \begin{equation}\label{eq:Krawtchouk_evolution}e^{-iJt}{\bf e}_0 =  \sum_{s= 0 }^{N} w(x_s)e^{-ix_st}\mathcal{X}(x_s)
 \end{equation}
where, for convenience, we have set
 $$w(x_s) = \frac{h_N}{\widehat{K}_N(s) K_{N+1}'(s)}, \quad s=0,1,2,\dots, N.$$
 In particular, we have 
\begin{equation}\label{eq:x0}
x_0(t) = (e^{-iJt}{\bf e}_0,{\bf e}_0)= \sum_{s = 0}^{N} w(x_s)e^{-ix_st}.
    \end{equation}
Since $J$ is persymmetric, we know that $\widehat{K}_N(s)=(-1)^{N+s}$ (see \cite[Corollary 6.2]{DMS} or \cite{VZh12}) and so we can write
\begin{equation}\label{eq:weights}w(x_s) = \frac{(-1)^{N+s} h_N}{K_{N+1}'(s)}, \quad s=0,1,2,\dots, N.
\end{equation}
Since
$$K_{N+1}(x) = \prod_{k=0}^{N}(x-k),$$
we have
\[K_{N+1}'(x) = \sum_{s=0}^{N}\prod_{k=0, k\neq s}^{N} (x-k)\] and hence  \[K_{N+1}'(s) = \prod_{k=0, k \neq s}^{N}(s-k).\]
By induction we get that 
$$\frac{(-1)^{N+s}}{K_{N+1}'(s)} = \frac{(-1)^{N+s}}{\prod_{k = 0, k \neq s}^{N} (s-k)} = \frac{1}{s! (N-s)!}$$
and thus we can write 
$$x_0(t) = \sum_{s = 0 }^{N} \frac{h_N}{s!(N-s)!} e^{-ix_st}.$$
Now, if $J$ has ESE then $x_0(t)=0$, that is,
$$ x_0(t) = \sum_{s = 0 }^{N} \frac{h_N}{s!(N-s)!} e^{-ix_st} = \sum_{s = 0 }^{N} \frac{h_N}{s!(N-s)!} e^{-i(s-N/2)t}=0,$$
multiplying the above relation by $\frac{N!}{h_N \cdot 2^N}$ yields 
$$\sum_{s = 0 }^{N} \frac{1}{2^N}{N\choose s} (e^{i\frac{t}{2}})^{N-s}(e^{-i \frac{t}{2}})^s= \cos^N\left(\frac{t}{2}\right)=0.$$
Therefore, we see that $\cos(\frac{t}{2}$) = 0. We now claim that this implies that 
\[
|x_N(t)|=|(e^{-iJt}{\bf e}_0,{\bf e}_N)|=1.
\]
Firstly, note that if $\cos(\frac{t}{2})=0$ then $\sin(\frac{t}{2}) = \pm1$. Then, similarly to what we did to get the expansion of ${\bf e}_0$, we obtain 
$${\bf e}_N = \sum_{s=0}^{N} (-1)^{N+s}\frac{\mathcal{X}(x_s)}{\|\mathcal{X}(x_s)\|^2}$$
and hence
 \begin{equation}\label{eq:xN}
 x_N(t) = (e^{-iJt}{\bf e}_0,{\bf e}_N) = \sum_{s=0}^{N} (-1)^{N+s} w(x_s) e^{-i x_s t}= \sum_{s=0}^N(-1)^{N+s}\frac{h_N}{s!(N-s)!}e^{-ix_st}.
 \end{equation}
Then since $e^{-iJt}$ is unitary, it follows from \eqref{eq:x0} that \[x_0(0)=\sum_{s=0}^Nw(x_s)=1\]
and so using the formula for $w(x_s)$ derived above, we have
\begin{equation*}
     h_N = \frac{1}{\sum_{s=0}^{N} \frac{1}{s!(N-s)!}}    
\end{equation*} so that 
\begin{equation*}
     \frac{h_N 2^N }{N!} = \frac{2^N}{\sum_{s=0}^{N} \frac{N!}{s!(N-s)!}}=1 
\end{equation*} and hence $h_N=\frac{N!}{2^N}$.
Therefore, multiplying \eqref{eq:xN} by $\left(\frac{-1}{i}\right)^N$, we have 
    $$\left(\frac{-1}{i}\right)^N x_N(t) = \sum_{s=0}^{N} \left(\frac{1}{2i}\right)^N {N \choose s} (-e^{-i\frac{t}{2}})^s (e^{i\frac{t}{2}})^{N-s} = 
    \sin^N\left(\frac{t}{2}\right) = (\pm1)^N.$$
Henceforth, for a persymmetric matrix with equidistant spectrum exhibiting PST, if there exists a time $t$ such that $x_0(t) = 0$ then  $|x_N(t)| = 1$. This directly implies early state exclusion is impossible. 
    \end{proof}

\section{Appearance of Early State Exclusion for Some Non-Equidistant Spectra}

In this section we construct a family of Jacobi matrices that realize PST and have ESE. To this end, we need to recall some basics of Chebyshev polynomials, which are defined by the relation
\begin{equation*}
T_n(x)=\cos(n \theta),\quad x=\cos\theta
\end{equation*}  
for $n=0,1,\dots$.
If the range of the variable $x$ is the interval $[-1,1]$, then the range of the corresponding variable $\theta$ can be taken as $[0, \pi]$. It is not so difficult to verify that they satisfy the orthogonality relation
\begin{equation}\label{eq:ChebyshevOrth}
\int_{-1}^{1}T_n(x)T_k(x)\frac{dx}{\sqrt{1-x^2}}=0, \quad n\ne k.   
\end{equation}
It is well known that all the zeroes of $T_n$ belong to $(-1,1)$. Moreover, it is also known that at most one zero of the quasi-orthogonal polynomial 
\[
T_{n+1}(x)+AT_n(x)
\]
of rank $1$ lies outside $(-1,1)$, see \cite[Chapter II, Theorem 5.2]{Ch78}. One can easily adapt one of the available proofs of this fact to a more general linear combination, which we will do here for the convenience of the reader.
\begin{lemma}\label{lemma:mzeroes}
    The quasi-orthogonal polynomial
\[
Q(x):=A_{n+1}T_{n+1}(x)+A_nT_n(x)+\dots+A_kT_k(x), \quad A_k\ne0, \quad k\leq n+1,
\]
has at least $k$ distinct zeros on $(-1,1)$.
\end{lemma}
\begin{proof}
Assume that $y_r$, $r=1,\dots, l$ are the points of the interval $(-1,1)$ at which $Q$ changes sign and also assume that $l<k$. Then the polynomial
\[
Q(x)\prod_{r=1}^l(x-y_r)
\]
does not change the sign on $(-1,1)$ and hence
\begin{equation}\label{eq:nonvanishingint}
 \int_{-1}^{1}Q(x)\prod_{r=1}^l(x-y_r)\frac{dx}{\sqrt{1-x^2}}\ne 0.
\end{equation}
On the other hand, since $\deg\displaystyle{\prod_{r=1}^l(x-y_r)}<k$ from \eqref{eq:ChebyshevOrth} we get
\[
\int_{-1}^{1}Q(x)\prod_{r=1}^l(x-y_r)\frac{dx}{\sqrt{1-x^2}}=\sum_{j=k}^{n+1}A_j\int_{-1}^{1}T_j(x)\prod_{r=1}^l(x-y_r)\frac{dx}{\sqrt{1-x^2}}=0,
\]
which contradicts \eqref{eq:nonvanishingint}. Therefore, $Q$ must change the sign at least $k$ times on $(-1,1)$.
\end{proof}
The above statement allows us to prove the following result.
\begin{theorem}
Let $N=2n-1$, where $n>1$ is a positive integer.  Then there exists a Jacobi matrix of order $N+1$ that realizes PST at time $T_0$ and has ESE. Moreover, for a given positive number $m$, it is possible to construct such a Jacobi matrix so that it has $m$ cases of ESE, i.e. there exist $m$ times $\{t_j\}_{j=1}^m$ where $x_0(t_j)=0$ and $t_j<T_0$.
\end{theorem}
\begin{proof}
Fix two positive integers $n>1$ and $m$, and define a finite sequence of numbers $\{\lambda_s\}_{s=0}^{2n-1}$ in the following way:
\begin{equation}\label{eq:uppereigenvalues}
\lambda_{n+k}=\frac{2m+2k+1}{2}, \quad k=0,1,\dots, n-1
\end{equation}
and 
\begin{equation}\label{eq:lowereigenvalues}
\lambda_{k}=-\lambda_{2n-1-k}, \quad k=0,1,\dots, n-1,
\end{equation}
where $m$ is a positive integer. In other words, this sequence has the property that the gaps between the eigenvalues are the same except the middle one:
\[
\lambda_{n}-\lambda_{n-1}=2m+1,\quad \lambda_{k+1}-\lambda_k=1, \quad k\ne n-1.
\]
We note that equations \eqref{eq:uppereigenvalues} and \eqref{eq:lowereigenvalues} give 
\[
\lambda_s=-\frac{2n+2m-2s-1}{2}, \quad s=0,1,2,\dots, 2n-1.\]
Next, by the Hochstadt theorem \cite[Theorem 3]{H74}, there is a unique persymmetric Jacobi matrix of order $2n$ whose eigenvalues are $\lambda_s's$. Evidently, the eigenvalues satisfy \eqref{PST_eigenvalues} with $T_0=\pi$ and thus $J$ realizes PST with the earliest transfer time $T_0=\pi$. This matrix also defines a finite family of orthogonal polynomials like in the case of Krawtchouk polynomials discussed in Section \ref{sec:absence of ESE}. In fact, following similar arguments, we can establish that
\begin{equation}
x_0(t) = (e^{-iJt}{\bf e}_0,{\bf e}_0)= \sum_{s = 0}^{2n-1} w(\lambda_s)e^{-i\lambda_st}
\end{equation}
with $w(\lambda_s)=w(\lambda_{2n-1-s})$ for $s=0,1,2,\dots, 2n-1$. As a result, we get 
$$x_0(t) = \sum_{s=0}^{n-1} 2 \cdot {w}(\lambda_s) \cdot \cos{\left((2n+2m-2s-1)\frac{t}{2}\right)},$$
which is a linear combination of Chebyshev polynomials. Namely, setting $x=\cos(t/2)$ the linear combination takes the form
$$x_0(t) = A_{2m+1}T_{2m+1}(x) + ... + A_{2n+2m-1}T_{2n+2m-1}(x),$$
where $A_{2m+1}=2w(\lambda_{n-1})\ne 0$. By Lemma \ref{lemma:mzeroes}, the linear combination as a polynomial in $x$ has at least $2m+1$ distinct zeroes on $(-1,1)$. In turn, it implies that $x_0(t)$, as a function of $t$, has $2m+1$ zeroes on $(0,2\pi)$. We know that one of them is at $t=T_0=\pi$ and since 
\[
|x_0(t+\pi)|=|x_0(t-\pi)|,
\]
we have that there are at least $m$ zeroes on $(0,\pi)$, meaning that there are $m$ cases of ESE.     
\end{proof}

\begin{remark}
  It is worth noting that one can establish existence of Jacobi matrices with ESE by applying a method of spectral surgery (see \cite{VZh12}) to Krawtchouk polynomials. Specifically, for an odd integer $N$, one may "remove" the two inner-most eigenvalues of an $(N+3) \times (N+3)$ Jacobi matrix corresponding to Krawtchouk polynomials and obtain a new $(N+1) \times (N+1)$ Jacobi matrix with initial state probability amplitude, $\tilde{x}_0(t)$, which is given in terms of the probability amplitude, $x_0(t)$, of an $(N+1) \times (N+1)$ Jacobi matrix corresponding to Krawtchouk polynomials. Namely, 
  \begin{equation}\label{eq:surgeryamp}
\tilde{x}_0(t)=\left(\left(\frac{N+1}{2}+1\right)\cos(t)-\left(\frac{N+1}{2}\right)\right)\cos^N\left(\frac{t}{2}\right), 
\end{equation} where $\cos^N\left(\frac{t}{2}\right)=x_0(t)$. Then we can see the presence of ESE from equation \eqref{eq:surgeryamp}. In particular, if set $N=3$ we get the matrix \eqref{4by4example} and the corresponding probability amplitude \eqref{eq:x_0for4by4}. Nevertheless, the proof that we propose in this section is very general and gives a great variety of outcomes. It only relies on the symmetry of the eigenvalues and a gap in the middle. It should be noted that the proof does not require the rest of the gaps to be uniform. 
\end{remark}

\noindent {\bf Acknowledgments.} The authors acknowledge the support of the NSF DMS grant 2349433. They are also indebted to Prof. Luke Rogers and Prof. Sasha Teplyaev for the opportunity to participate in UConn REU site in Summer 2024 and for many helpful and productive discussions. M.D. was also partially supported by the NSF DMS grant 2008844. Finally, M.D. is extremely grateful to Prof. Christino Tamon for bringing the ESE problem to his attention.

\end{document}